\documentclass[a4paper,11pt]{amsart}
\usepackage{amsfonts,amssymb,epsfig,latexsym}
\usepackage{amsmath}
\usepackage[latin1]{inputenc}
\usepackage{color,layout}
\usepackage{ifthen}

\usepackage{pdfsync}
\usepackage{graphicx}
\usepackage{color}

\newtheorem{theorem}{Theorem}[section]
\newtheorem{lemma}[theorem]{Lemma}
\newtheorem{proposition}[theorem]{Proposition}
\newtheorem{definition}[theorem]{Definition}
\newtheorem{remark}[theorem]{Remark}
\newtheorem{corollary}[theorem]{Corollary}

\def\la{\left\langle}
\def\ra{\right\rangle}

\def\square{\hbox{\vrule\vbox{\hrule\phantom{o}\hrule}\vrule}}

{\vspace{2ex}\noindent{\sl Proof: }}{{\null\hfill\square\bigskip}}

\def\re{{\rm Re}}
\def\R{\mathbb {R}}

\def\C{\mathbb {C}}
\def\N{\mathbb {N}}

\def\re{\mathop{\rm Re}\nolimits}
\def\im{\mathop{\rm Im}\nolimits}
\def\la{\langle}
\def\ra{\rangle}
\def\cH{\mathcal {H}}
\def\cD{\mathcal {D}}

\def\cG{\mathcal {G}}
\def\cF{\mathcal {F}}
\def\cO{\mathcal {O}}
\def\cW{\mathcal {W}}
\def\cM{\mathcal {M}}
\def\cR{\mathcal {R}}

\newcommand{\be}{\begin{equation}}
\newcommand{\ee}{\end{equation}}
\newcommand{\beee}{\begin{eqnarray*}}
\newcommand{\eeee}{\end{eqnarray*}}
\newcommand{\bee}{\begin{eqnarray}}
\newcommand{\eee}{\end{eqnarray}}

\numberwithin{equation}{section}

\begin{document}
\title[Resonances in the Born-Oppenheimer approximation]{On the Born-Oppenheimer approximation of diatomic molecular resonances}

\newcommand{\printdate}{\@date}

\author {A. Martinez${}^1$\, \&\, V. Sordoni${}^2$}

\maketitle

\addtocounter{footnote}{1}
\footnotetext{{\tt\small Universit\`a di Bologna,  
Dipartimento di Matematica, Piazza di Porta San Donato, 40127
Bologna, Italy, 
vania.sordoni@unibo.it }} 
\addtocounter{footnote}{1}
\footnotetext{{\tt\small Universit\`a di Bologna,  
Dipartimento di Matematica, Piazza di Porta San Donato, 40127
Bologna, Italy, 
andre.martinez@unibo.it }}  

\begin{abstract}
We give a new reduction of a general diatomic molecular Hamiltonian, without modifying it near the collision set of nuclei. The resulting effective Hamiltonian is the sum of a smooth semiclassical pseudodifferential operator (the semiclassical parameter being the inverse of the square-root of the nuclear mass), and a semibounded operator localised in the elliptic region corresponding to the nuclear collision set. We also study its behaviour on exponential weights, and give several applications where molecular resonances appear and can be well located.
\end{abstract}
\pagebreak

\section{Introduction}

This paper is devoted to the Born-Oppenheimer reduction of a diatomic molecular Hamiltonian, near energy levels where resonances may appear.

The principle of the Born-Oppenheimer reduction goes back to 1927 with the work \cite{BoOp}, where the fact that the nuclei are much heavier than the electrons is exploited in order to approximate the complete molecular Schr\"odinger operator by a reduced Hamiltonian, acting on the positions of the nuclei only, and in which the electrons are involved through the effective electric potential they create only. This principle has been widely used by chemists since that period, but the mathematically rigorous justifications of this reduction are much more recent. They started with \cite{CDS}, where it has been justified up to error-terms of order $h^2:=M_{nucl}^{-1}$ ($M_{nucl}$ being the average mass of the nuclei), and continued with \cite{Ha1, Ma1} (up to ${\mathcal O}(h^\infty)$, but concerning smooth interactions only), \cite{Ha2} (up to ${\mathcal O}(h^\infty)$ for bounded states of diatomic molecules), \cite{KMSW} (exact reduction for bounded states of polyatomic molecules), \cite{Ha3, Ha4, MaSo1, MaSo2} (up to ${\mathcal O}(h^\infty)$ for the quantum evolution problem). 
Let us also mention results on the Born-Oppenheimer reduction for the scattering process \cite{Ra, KMW1, KMW2}.

Concerning the reduction for resonant (or metastable) states with singular  (Coulomb-type) interactions, to our knowledge it is treated in \cite{MaMe} only. In that paper, following ideas from \cite{KMSW}, a regularisation of the Hamiltonian is constructed far from the collision set of the nuclei, and this gives rise to an effective Hamiltonian of pseudodifferential type. However, this effective Hamiltonian is not constructed for the exact molecular Hamiltonian, but for a modified one where the singularity coming from the collision set of the nuclei has been artificially removed. Because of this, additional assumptions have to be done in order to be able to compare the resonances obtained from the effective Hamiltonian to those of the original Hamiltonian.

Indeed, the main inconvenient  of \cite{MaMe} lies in some precise assumption (see condition (ii) of Proposition 6.1 in \cite{MaMe}) that involves the unmodified operator directly, and appears to be difficult to verify in general.

Here, we  construct an effective Hamiltonian for the unmodified molecular Schr\"odinger operator, in such a way that the contribution of the collision set of the nuclei is clearly individuated and separated from its complementary. The resulting effective Hamiltonian appears to be the sum of a semiclassical pseudodifferential operator (where, as usual, the semiclassical parameter is the inverse of the square-root of the nuclear mass), and a semibounded operator localised near the collision set of the nuclei. Thanks to this localisation it is possible to apply a general technique (originated in \cite{HeSj2}) in order to compare the resonances of the full effective Hamiltonian to those of its pseudodifferential part.

In the next section, we prove an abstract result of reduction, similar to the Feshbach standard result, but with the difference that the nuclei position-space is split in two parts, each of them giving rise to separated contributions in the final effective Hamiltonian (see Theorem \ref{abstractth}). In Section \ref{sect3}, we apply this result to the particular case of a diatomic molecular Hamiltonian, with Coulomb singularities. Then, in Section \ref{sect4}, we give a representation of the effective Hamiltonian in terms of a matrix operator, that can be split into a smooth pseudodifferential part and an operator localised into the elliptic region (see Theorem \ref{Thpseudo}). In addition, the conjugation of these operators by an exponential weight-function is studied, too. This is used in Section \ref{sect5} to investigate their action on WKB solutions. In Section \ref{sectellipt}, a general method (taken from \cite{HeSj2}) is described, in order to compare the resonances of the full effective Hamiltonian to those of its pseudodifferential part. Finally, in Section \ref{sect7}, we give a list of applications where our result, together with other techniques, permits to locate molecular resonances and give estimates on their widths.

\section{An abstract splitting result}
\label{sect2}
Let $\cH_Q$ be a Hilbert space. We consider a general (unbounded) closed operator $P$ with domain $\cD_P$ on $\cH:=L^2(\R^n)\otimes\cH_Q$, that can be written on the form,
$$
P=K_0\otimes I +\int^{\oplus}_{\R^n} Q(x) dx=:K_0\otimes I +Q,
$$
where $K_0$ is a closed operator on $L^2(\R^n)$ with dense domain $\cD_0$, such that,
\be
\label{K0pos}
\re K_0 \geq 0,
\ee 
and, for almost all $x\in \R^n$, $Q(x)$ is a lower semi-bounded closed operator on $\cH_Q$, with lower semi-bound $-C_0$ and dense domain $\cD_Q$, both independent of $x$. In particular, one has,
$$
\cD_P = \cD_0\otimes\cD_Q.
$$
We are interested in the spectrum of $P$ near some real interval $I_0:=(-\infty, \lambda_0]$,  and we assume the existence of two open subsets $\cW_0$ and $ \cW_1$ of $\R^n$, together with a continuous family of projections $(\Pi (x))_{x\in\R^n}$ on $\cH_Q$ and some $\delta_0\geq \delta_1>0$, such that,
\be
\label{assabstr}
\begin{split}
&\cW_0 \cup \cW_1 =\R^n;\\
&\re Q(x) \geq \lambda_0 + \delta_0\quad a.\, e. \, \mbox{ on } \cW_0;\\
&\Pi (x)\, :\, \cD_Q \to \cD_Q;\\
&[Q(x),\Pi (x)]=0 \quad a.\, e. \, \mbox{ on } \cW_1;\\
&\re (Q(x)-\lambda_0- \delta_1)(1-\Pi(x)) \geq 0\quad a.\, e. \, \mbox{ on } \cW_1.
\end{split}
\ee
We set
$$
\Pi := \int_{\R^n}^{\oplus}\Pi (x) dx\quad ;\quad \cH_{\rm red}:= \Pi (\cH)\subset \cH,
$$
 Our aim is to reduce the spectral study of $P$ near $I_0$, to that of an operator acting on the `reduced' space $ \cH_{\rm red}$ (as in the standard Feshbach method), but in such a way that the contributions of $Q_0:=\int^{\oplus}_{\cW_0} Q(x) dx$ and $Q_1:=\int^{\oplus}_{\cW_1} Q(x) dx$ are clearly individuated and separated (the idea is that, in the applications, $Q_1$ is a ``smooth'' operator, while $Q_0$ is ``singular'' but elliptic).

Let $\varphi_0, \varphi_1, \psi_0,\psi_1\in C^\infty (\R^n, [0,1])$ be cut-off functions such that,
\be
\label{partfunct}
\begin{split}
& {\rm Supp}\, \varphi_j \cup  {\rm Supp}\, \psi_j \subset \cW_j\quad (j=0,1);\\
& \varphi_0^2 +\varphi_1^2 =1\, \mbox{ on }\, \R^n;\\
& \psi_j =1\, \mbox{ on }\,{\rm Supp}\, \varphi_j\quad (j=0,1).
\end{split}
\ee
We set,
\be
\label{opmod}
\begin{aligned}
& P_0=K_0+Q_0:= K_0 + Q + (\lambda_0+\delta_0+C_0)(1-\psi_0);\\
& P_1=K_0+Q_1:= K_0 + Q\psi_1 + (\lambda_0+\delta_0)(1-\psi_1);\\
&\hat\Pi := I -\Pi.
\end{aligned}
\ee
In particular, one has $\varphi_j P_j= \varphi_j P$ and $P_j\varphi_j = P\varphi_j$ for $j=0,1$. Moreover, by construction, one also has,
\be
\begin{split}
&\re P_0 \geq \lambda_0+\delta_0;\\
&\re \hat\Pi (P_1-\lambda_0-\delta_1) \hat\Pi\geq 0,
\end{split}
\ee
and thus, for $z$ in a small enough complex neighborhood of $I_0$, both $P_0-z$ and the restriction of $\hat\Pi P_1 \hat\Pi-z$ to the range of $\hat\Pi$ are invertible, with bounded inverse.
We set,
\be
\label{defbasic}
\begin{split}
&X_0=X_0(z):=\hat\Pi(P_0-z)^{-1}\hat\Pi;\\
&X_1=X_1(z):=\hat\Pi \left( \hat\Pi(P_1-z)\hat\Pi\right)^{-1}\hat\Pi;\\
& M_j:= [P_j,\Pi]\quad (j=0,1);\\
& T_j:= [K_0,\varphi_j]\quad (j=0,1);\\
& Y=Y(z):= \varphi_0X_0T_0 +\varphi_1X_1T_1;\\
& Y'=Y'(z):= T_0X_0\varphi_0 +T_1X_1\varphi_1;\\
& Y_1:=\varphi_0X_0M_0\varphi_0+\varphi_1X_1M_1\varphi_1;\\
& Y_2:=\varphi_0M_0X_0M_0\varphi_0+\varphi_1M_1X_1M_1\varphi_1;\\
& Y_3:=\varphi_0M_0X_0T_0+\varphi_1M_1X_1T_1;\\
& Y_4:=\varphi_0M_0\varphi_0+\varphi_1M_1\varphi_1.
\end{split}
\ee
Observe that $Y$ and $Y'$ are bounded operators on $\cH$ and, in the applications, they will actually be very small. Our result is,
\begin{theorem}\sl 
\label{abstractth}
Assume $||Y(z)||<1$ and $||Y'(z)||<1$. Then, for $z$ in a small enough complex neighborhood of $I_0$, one has the equivalence,
$$
z\in \sigma (P) \Longleftrightarrow 0\in \sigma (A(z)),
$$
where,
$$
A(z)=\Pi \left(z-P  \right)\Pi +B\, :\, \cH_{\rm red}\cap \cD_P \to \cH_{\rm red},
$$
with,
$$
B(z):=\Pi(-Y_2+(M_0+Y_3-Y_4)(1+Y)^{-1}(1-Y_1))\Pi.
$$
\end{theorem}
\begin{proof}
For $z\in\C$ near $I_0$, we consider the Grushin problem,
\be
\label{grushin}
\cG(z):=\left(\begin{matrix}
P-z & I \\
\Pi & 0
\end{matrix}\right)\, :\, \cD_P\oplus \cH_{\rm red} \to \cH\oplus \cH_{\rm red}.
\ee
We also consider,
\be
\cG_j(z):=\left(\begin{matrix}
P_j-z & I \\
\Pi & 0
\end{matrix}\right)\quad\quad (j=0,1).
\ee
It is straightforward to check that $\cG_j(z)$ ($j=0,1$) invertible, with  inverse given by,
\be
\cG_j(z)^{-1}:=\left(\begin{matrix}
X_j & I-X_jM_j \\
\Pi(1+M_jX_j) & \Pi(z-P_j-M_jX_jM_j)
\end{matrix}\right).
\ee
Setting $\cF(z) := \varphi_0\cG_0^{-1}\varphi_0 + \varphi_1\cG_1^{-1}\varphi_1$ and $U:=\varphi_0X_0\varphi_0+\varphi_1X_1\varphi_1$, we find,
\be
\cF(z):=\left(\begin{matrix}
U & 1-Y_1 \\
\Pi(1 + Y_1') &\Pi (z-P-Y_2)
\end{matrix}\right),
\ee
with,
$$
Y_1':=\varphi_0M_0X_0\varphi_0+\varphi_1M_1X_1\varphi_1,
$$
and then, using that $\varphi_0[P,\Pi]=\varphi_0M_0$, $\varphi_1[P,\Pi]=\varphi_1M_1$, $U\Pi=0$,  $Y_1'\Pi =0$, and $Y_4\hat\Pi =\Pi Y_4$,
$$
\cF(z)\cG(z):=\left(\begin{matrix}
I+Y & 0 \\
G_1& I
\end{matrix}\right),
$$
with,
$$
G_1:=-\Pi M_0-Y_3+\Pi Y_4.
$$
Therefore, the operator,
\be
\label{invG}
\left(\begin{matrix}
1+Y & 0\\
G_1 & I
\end{matrix}\right)^{-1}\cF(z)=\left(\begin{matrix}
(I+Y)^{-1} & 0 \\
-G_1(I+Y)^{-1} & I
\end{matrix}\right)\cF(z)
\ee
is a left-inverse for $\cG(z)$.

In the same way, using that $\Pi Y_2\Pi =Y_2$, we also find,
$$
\cG(z)\cF(z):=\left(\begin{matrix}
1+Y' & G_2 \\
0 & I
\end{matrix}\right),
$$
with,
$$
G_2:=M_0-T_0X_0M_0\varphi_0-T_1X_1M_1\varphi_1-Y_4\Pi.
$$

This proves that $\cG(z)$ is surjective, too, and thus invertible with inverse given by (\ref{invG}). Moreover, if $A(z)$ stands for the coefficient $(2,2)$ of  $\cG(z)^{-1}$, one has the standard algebraic property,
$$
z\in\sigma (P) \Longleftrightarrow 0\in\sigma (A(z)).
$$
By definition, we also have $A(z)=\Pi (z-P-Y_2)-G_1(1+Y)^{-1}(1-Y_1)$,
and the result follows.
\end{proof}
\begin{remark}\sl
If we neglect all the terms involving $Y$, $M_1=[K_0,\Pi]$, $T_1$ or $T_2$ (that will all be small in the applications we have in mind), we see that the operator $A(z)$  reduces to its principal part $A_0(z)$, given by,
$$
A_0(z):= \Pi (z-P-RX_0R)\Pi,
$$
with $R:=[Q,\Pi] =\varphi_0[Q,\Pi]\varphi_0$.
\end{remark}

\section{Diatomic molecular resonances}
\label{sect3}

\subsection{The model}
\label{sect3.1}

We consider the selfadjoint  operator $H$ on $L^2(\R^3_x\times \R^{3p}_y)$, with domain $H^2(\R^3_x\times \R^{3p}_y)$ defined as 
\be
\label{eq:def} 
\begin{split} 
&H=-h^2\Delta+ H_{\rm el}(x)\\
&H_{\rm el}(x)=\tilde H_{\rm el}(x)+\frac{\alpha}{\vert x\vert}\\
&\tilde H_{\rm el}(x)=-\Delta_y+V(x,y)
\end{split}
\ee
with 
$$V(x,y)=\sum_{j=1}^n\left(\frac{\alpha_j^+}{\vert y_j+x\vert}+\frac{\alpha_j^-}{\vert y_j-x\vert}\right)+
\sum_{j,k=1}^n\frac{\alpha_{jk}}{\vert y_j-y_k\vert}$$
where $\alpha, \alpha_j^{\pm}$ and $\alpha_{jk}$ are real constant and $\alpha>0$, $\alpha_j^{\pm}<0$.

In this model, $x$ stands for the relative position of the nuclei, $y$ for the position of the electrons, and $h^2$ for the ratio between the electronic and nuclear masses. In particular, $H_{\rm el}$ is the electronic Hamiltonian with electronic mass normalized at $m=\frac12$.
 
Let us define the resonances of $P$ by using the analytic distortion introduced in  \cite{Hu}.\\
Let $\omega\, :\, \R^3 \rightarrow \R^3$ be a smooth  odd vector field  such that 
\begin{itemize}
\item $\omega(x) = 0$  for $|x|\leq R$ ($R>0$ large enough);
\item $\omega(x) = x$  for $|x| >>1$,
\item For any rotation $\cR$ on $\R^3$, one has  $\omega (\cR x) =\cR \omega (x)$.
\end{itemize}
(In other words, we take $\omega (x)$ on the form $\omega (x) = \chi (|x|)x$ with $\chi (t) =0$ when $t\leq R$, and $\chi (t) =1$ for $t>>1$.) For $\mu\in\R$ small enough, we consider the transformation
$$F_\mu(x,y)=(x+\mu \omega(x),y_1+\mu \omega(y_1),\dots,y_p+\mu \omega(y_p) )$$
and the analytic distorsion  $U_{\mu}$ associated to $F_\mu$  
defined as 
$$
U_{\mu}\phi(x,y)=\phi(F_\mu(x,y) ).
$$
Then the family,
$$H_\mu=U_{\mu}H U_{\mu}^{-1}$$
can be extended to  small complex values of $\mu$, and  we can give the following definition:
\begin{definition}A complex number $\rho $ is a resonance of $H$ if $\re \rho>\inf\sigma_{ess}(H)$ and there exists $\mu $ small enough, with $\im \mu>0$, such that $\rho\in\sigma_{disc}(H_\mu)$.
\end{definition}
\begin{remark}\sl
\label{rmmubar}
Using the self-adjointness of $H$, one can prove that this is also equivalent to $\overline \rho \in \sigma_{disc}(H_{\overline\mu})$.
\end{remark}
In the following we denote by $\Gamma (H)$ the set of such resonances.

\subsection{General assumptions and reduction}
\label{sect3.2}

We assume that, for some fixed $\lambda_0\in\R$, and for all  $x\in \R^3$, one has,
\be
\sigma(\tilde H_{\rm el}(x))\cap(-\infty, \lambda_0]\quad {\rm is \;discrete}.
\ee
Moreover, we also assume the existence of some finite $m\geq 1$, such that
\be
 \#\,\sigma(\tilde H_{\rm el}(x))\,\cap\,]-\infty, \lambda_0]\leq N.
 \ee
Let us denote by 
$$\tilde\lambda_1(x)<\tilde\lambda_2(x)\leq \dots\leq\tilde \lambda_N(x)$$ the first $N$ eigenvalues of $H_{\rm el}(x)$ and assume there exists a gap between them and the rest of the spectrum of $H_{\rm el}(x)$, that is, there exists some $\delta>0$, such that, 
\be
\inf_{x\in\R^3}{\rm dist}\left(\sigma(\tilde H_{\rm el}(x)\setminus\{\tilde\lambda_1(x), \dots, \tilde\lambda_N(x)\}, \{\tilde\lambda_1(x), \dots, \tilde\lambda_N(x)\}\right)\geq \delta
\ee
This fact implies that the spectral projection $\Pi_{\rm el}(x)$ of $\tilde H_{\rm el}(x)$ associated to $\{\tilde\lambda_1(x), \dots, \tilde\lambda_N(x)\}$ is $C^2$ with respect to $x\in \R^3$ (see \cite{CoSe, CDS}).

In the following, we set
$$\lambda_j(x):=\tilde\lambda_j(x)+\frac{\alpha}{\vert x\vert}, \quad\quad j=1,\dots,N.$$
Since $\alpha_{\pm}<0$ then there exists $C>0$ such that
$$\lambda_N(x)\leq C+\frac{\alpha}{\vert x\vert}.$$

For $x\not=0$, we set
$$\tilde H_{\rm el}^\mu(x)=U_\mu \tilde H_{\rm el}(x+\mu\omega(x))U_\mu^{-1}$$
and
$$H_{\rm el}^\mu(x)=\tilde H_{\rm el}^\mu(x)+\frac{\alpha}{\vert x+\mu\omega(x)\vert}.$$

By Lemma 2.1 in \cite{MaMe}, we also know that there exists $C_1>0$ such that, for all $x\not=0$ 
$$\lambda_1(x)\geq \frac{\alpha}{\vert x\vert }-C_1$$

For $x\in \R^n$, let $\gamma (x)$ be a continuous family of simple loops of $\C$, enclosing $\{\tilde \lambda_j(x)\, ;\, j=1,\dots,N\}$
and having the rests of $\sigma(\tilde H_{\rm el}(x))$ in its exterior.

By the gap condition, we may assume that 
$$\min_{x\in \R^3}{\rm dist}( \gamma(x), \sigma(\tilde H_{\rm el}(x)))\geq \frac{\delta}{2}>0$$
Moreover, $\gamma(x)$ can be taken in some fix compact set of $\C$.\\
Thanks to Lemma 2.3 of \cite{MaMe}, if $\mu\in\C$ is small enough, then for any $x\in \R^3$ and $z\in \gamma(x)$, the operator 
$(z-\tilde H_{\rm el}^\mu(x))^{-1}$ exists and satisfies
$$(z-\tilde H_{\rm el}^\mu(x))^{-1}-(z-\tilde H_{\rm el}(x))^{-1}=\mathcal O(\vert\mu\vert)$$
uniformly.
Then,  for $\mu\in \C$ sufficiently small, we  can define,
$$\Pi_{\rm el}^\mu(x)=\int_{\gamma(x)}(z-\tilde H_{\rm el}^\mu(x))^{-1}\;dx.$$

At that point, we fix $\mu =i\mu'$ with $\mu' >0$ small enough, and we apply Theorem \ref{abstractth} with,
\begin{itemize}
\item $\cH_Q:= L^2(\R^p)$;
\item $K_0:= -h^2 U_\mu \Delta U_\mu^{-1}=[(1+\mu{}^td\omega(x))^{-1}hD_x]^2$;
\item $Q(x):=H_{\rm el}^\mu(x)$;
\item $\cD_Q:= H^2(\R^p)$;
\item $P:=H_\mu$;
\item $\cD_P:= H^2(\R^{3+p})$;
\item $\cW_0:=\{ |x|< 2\delta_1\}$ with $\delta_1>0$ arbitrarily small;
\item $\cW_1:=\{ |x|>\delta_1\}$;
\item $\Pi(x)= \Pi_{\rm el}^\mu(x)$.
\end{itemize}
We observe that all the properties (\ref{assabstr}) are satisfied (with, indeed, $\delta_0$ arbitrarily large), and in addition, endowing $H^s(\R^3)$ with the semiclassical norm $|| u||_{H^s} :=||h^{-n/2}\la \xi \ra^s \hat u(\xi /h)||_{L^2}$ (where $\hat u$ stands for the usual Fourier transform),
we see that,
$$
X_0, X_1 =\cO (1)\, :\, H^{-1}(\R^3 ; L^2(\R^p))\to H^1(\R^3 ; L^2(\R^p)),
$$
and thus, we have,
\be
Y, Y'=\cO(h)\, :\, H^{-1}(\R^3 ; L^2(\R^p)) \to L^2(\R^3 ; L^2(\R^p)).
\ee
Moreover, using the fact that, for any function $f$, one has,
$$
\left((I+\mu {}^td\omega (x))^{-1}\nabla_x -(I+\mu {}^td\omega (y))^{-1}\nabla_y\right)f(x+\mu \omega (x) -y-\mu \omega (y))= 0,
$$
we see that we can use the same argument as in \cite{CoSe, CDS}, and conclude that $\Pi_{\rm el}^\mu (x)$ is $C^2$ with respect to $x\in\R^3$. As a consequence, and since $K_0$ is a 0-th order semiclassical differential operator of degree 2 with respect to $x$, and $[Q,\Pi]=0$ everywhere, we also have,
\be
\begin{split}
&M_0=M_1=\cO (h) \, :\, L^2(\R^3 ; L^2(\R^p)) \to H^{-1}(\R^3 ; L^2(\R^p));\\
&M_0=M_1=\cO (h) \, :\, H^1(\R^3 ; L^2(\R^p)) \to L^2(\R^3 ; L^2(\R^p)).
\end{split}
\ee
We also deduce that $||Y_2||+||Y_3|| =\cO(h^2)$, and since $\Pi M_0\Pi =0$ and $\Pi Y_4\Pi =0$,
in that case Theorem \ref{abstractth} becomes,
\begin{theorem}\sl
For $h>0$ small enough and $z$ in a small enough complex neighborhood of $(-\infty, \lambda_0]$, one has the equivalence,
$$
z\in \Gamma(H) \Longleftrightarrow 0\in \sigma (A_\mu(z)),
$$
where,
$$
A_\mu(z)=\Pi_{\rm el}^\mu \left(z-H_\mu  \right)\Pi_{\rm el}^\mu +B_\mu(z)\, :\, \cH_{\rm red}\cap \cD_P \to \cH_{\rm red},
$$
with,
$$
B_\mu(z)=\Pi_{\rm el}^\mu\left(-Y_2+(M_0+Y_3-Y_4)(1+Y)^{-1}(1-Y_1)\right)\Pi_{\rm el}^\mu=\cO(h^2).
$$
\end{theorem}
\begin{remark}\sl
Using Remark \ref{rmmubar}, we see that this is also equivalent to: $0\in \sigma(A_{\overline\mu}(\overline z))$.
\end{remark}
\begin{remark}\sl
In particular, the principal part of $A_\mu (z)$ is given by,
$$
A_\mu^0 (z):= \Pi_{\rm el}^\mu\left(z-H_\mu \right)\Pi_{\rm el}^\mu=U_\mu \Pi_{\rm el}\left(z-H \right)\Pi_{\rm el}U_\mu^{-1}.
$$
\end{remark}

\section{Smooth representation of the effective Hamiltonian}
\label{sect4}

Now, we are interested in the structure  of the effective Hamiltonian $A_\mu (z)$, and in particular in its possible representation as a semiclassical pseudodifferential operator, at least away from $x=0$. In view of the study of its action on  WKB-type functions, we also study the conjugated operators $e^{s/h}A_\mu (z)e^{-s/h}$ for convenient functions $s=s(x)$.

Since $\Pi(x)$ depend continuously on $x\in\R^3$ and one can find $m$ continuous section $v_1(x),\dots, v_N(x)$ generating ${\rm Ran}\Pi(x)$ and we can also assume that they form an orthonormal family.
Moreover, one can easily check that
$\tilde\lambda_1(x),\tilde\lambda_2(x),\dots,\tilde \lambda_N(x)$  depend on $\vert x\vert$ only,
and can be reindexed in such a way that each of the them depends analytically on $x\not=0$.\\
The arguments of \cite{MaMe} show that one can construct  
 finite family of bounded open sets $(\Omega_j)_{0\leq j\leq J}$ in $\R^n$, with $\Omega_0\subset\{\psi_1=0\}$, and  a corresponding family of unitary operators $U_j(x)$ ($j=0,\cdots ,J$; $x\in \Omega_j$),  with $U_0=1$ such
that (denoting by $U_j$ the unitary operator on $L^2(\Omega_j;{L^2(\R^{3p})})\simeq L^2(\Omega_j)\otimes L^2(\R^{3p})$ induced by the action of $U_j(x)$ on ${L^2(\R^{3p})}$),
\begin{itemize}
\item $\R^{3p} =\cup_{j=0}^J\Omega_j$;
\item For all $j=0,\cdots, J$ and $x\in  \Omega_j$, $U_j(x)$ leaves $H^2(\R^{3p})$ invariant;
\item For all $j$, the operator 
$U_j (-h^2\Delta_x)U_j^{-1}$
 is a semiclassical differential operator with operator-valued symbol, of the form,
$$
 U_j  (-h^2\Delta_x)U_j^{-1} = (-h^2\Delta_x)+h\sum_{|\beta|\leq 1}\omega_{\beta ,j}(x;h)(hD_x)^\beta,
$$
 where $\omega_{\beta ,j}(-\Delta_y+1)^{\frac{|\beta|}{2} -1}\in C^\infty (\Omega_j; {\mathcal  L}(L^2(\R^{3p})))$ for any $\gamma\in\N^n$ (here, ${\mathcal  L}(L^2(\R^{3p}))$ stands for the Banach space of bounded operators on $L^2(\R^{3p})$), and the quantity  $\Vert\partial_x^\gamma \omega_{\beta ,j}(x;h) (-\Delta_y+1)^{\frac{|\beta|}{2} -1}\Vert_{{\mathcal  L}(L^2(\R^{3p}))} $ is bounded uniformly with respect to $h$ small enough and locally uniformly with respect to $x\in\Omega_j$;
\item For all $j$, the operators $U_j(x)Q(x)\psi_1U_j(x)^{-1}$ and $U_j(x)(-\Delta_y + 1)U_j(x)^{-1}$ are in $C^\infty (\Omega_j; {\mathcal  L}(H^2(\R^{3p}),L^2(\R^{3p}))$ ;
\end{itemize}
In particular, following the terminology of \cite {MaSo2}, one can check that  the corresponding operator $P_1$ (defined as in (\ref{opmod}))  is a twisted pseudodifferential operator on $\R^3$ associated with $(\Omega_j, U_j)_{j=0,\dots, J}$.

Moreover, if we also assume that the $\tilde \lambda_j$'s are non degenerate and separated at infinity, in the sense that there exists $C>0$ such that,
\be
\label{sepationinf}
\inf_{j\not= k}\vert \tilde\lambda_j(x)-\tilde\lambda_k(x)\vert\geq \frac{1}{C}, \quad{\rm for }\;\vert x\vert\geq C,
\ee
then, by Proposition 5.1 of \cite{MaMe}, for $\mu\in \C$ small enough, there exist $m$ functions $w_{k,\mu}(x,y)\in C^0(\R^2; H^2(\R^{3p}))$, $k=1,\dots N$, depending analytically on $\mu$ near 0,  such that 
\be
\label {wkmutwist}
\begin{split}
& \langle w_{k,\mu},  w_{\ell,\bar\mu}\rangle_{L^2(\R^{3p})}=\delta_{k,\ell};\\
& \mbox{For } x\in \cW_1, \,( w_{k,\mu})_{1\leq k\leq N} \mbox{ form a basis of } {\rm Ran}\Pi_{\rm el}^\mu(x);\\
&  w_{k,\mu}\in C^{\infty}( \cW_0, H^2(\R^{3p});\\
& \mbox{For }  \vert x\vert  \mbox{ large enough, }  w_{k,\mu} \mbox{ is an eigenfunction of } Q_\mu(x)\\
\hskip 2cm & \mbox{ associated with }\lambda_k(x+\mu\omega(x));\\
& \mbox{For } j=1,\dots, J,\,
U_j(x) w_{k,\mu}\in C_b^{\infty}(\Omega_j, H^2(\R^{3p}).
\end{split}
\ee
For $u\in L^2(\R^{n+p})$ and $x\in\R^n$, we set,
$$
\widetilde\Pi_{\rm el}^\mu(x)u:=\sum_{k=1}^N \langle u,  w_{k,\bar\mu}\rangle_{L^2(\R^{3p})} w_{k,\mu}.
$$
In particular, for $x\in \cW_1$, one has $\widetilde\Pi_{\rm el}^\mu(x)=\Pi_{\rm el}^\mu(x)$, and we also observe that, if $\delta_1$ has been chosen small enough (in the definition of $\cW_0$), then, for $x\in\cW_0\backslash \{0\}$, one has $\re Q(x)\geq \lambda_0+\delta_0$. As a consequence, all the properties (\ref{assabstr}) are satisfied with $\Pi (x):=\widetilde\Pi_{\rm el}^\mu(x)$, too.

In the following we set
$$R^{-}_\mu:\mathop \oplus_1^N L^2(\R^{3})\rightarrow L^2(\R^{3+3p}), R^{-}_\mu(u_1^-,\dots, u_N^-)=\sum_{k=1}^N u_k^- w_{k,\mu}$$
and 
$$R^{+}_\mu=(R^{-}_\mu)^*:L^2(\R^{3+3p})\rightarrow\mathop \oplus_1^N L^2(\R^{3}), R^{+}_\mu g=\mathop \oplus_1^N\langle g,  w_{k,{\bar\mu}}\rangle_{L^2(\R^{3p})},$$
so that we have,
$$
R^+_\mu R^-_\mu = I \quad;\quad R^-_\mu R^+_\mu =\widetilde\Pi_{\rm el}^\mu=\Pi.
$$
As a consequence, $R_\mu^+$ sends isomorphically $\cH_{\rm red}$ into $L^2(\R^3)^{\oplus N}$, with inverse $R_\mu^-$. Moreover, by construction we also see that $R_\mu^+$ sends $H^2(\R^{3+3p})$ into $H^2(\R^3)$. Thus, in this case the study of the Grushin operator introduced in (\ref{grushin}) is equivalent to that of,
$$
\begin{aligned}
\widetilde \cG(z):&=
\left(\begin{matrix}
I & 0 \\
0 & R_\mu^+
\end{matrix}\right)\left(\begin{matrix}
P-z & I \\
\Pi & 0
\end{matrix}\right)
\left(\begin{matrix}
I & 0 \\
0 & R_\mu^-
\end{matrix}\right)\\
& =\left(\begin{matrix}
P-z & R_\mu^- \\
R_\mu^+ & 0
\end{matrix}\right)\, :\, \cD_P\oplus L^2(\R^3)^{\oplus N} \to \cH\oplus H^2(\R^3)^{\oplus N},
\end{aligned}
$$
and Theorem \ref{abstractth} gives us the first assertion of the following result :
\begin{theorem}\sl
\label{Thpseudo} Let $\mu=i\mu'$ with $\mu'>0$ fixed small enough.
For $h>0$ small enough and $z$ in a small enough complex neighborhood of $(-\infty, \lambda_0]$, one has the equivalence,
\be
\label{feschb}
z\in \Gamma(H) \Longleftrightarrow 0\in \sigma (\widetilde A_\mu(z)),
\ee
where,
$$
\widetilde A_\mu(z)=R_\mu^+ \left(z-H_\mu +\widetilde B_\mu(z) \right)R_\mu^-\, :\, H^2(\R^3)^{\oplus N} \to L^2(\R^3)^{\oplus N},
$$
with,
$$
\begin{aligned}
\widetilde B_\mu(z)&=-Y_2+(M_0+Y_3-Y_4)(1+Y)^{-1}(1-Y_1)\\
&=-RX_0R+\cO(h).
\end{aligned}
$$
Here, $R:=[Q,\Pi]=\varphi_0[Q,\Pi]\varphi_0$.

Moreover,  if $s=s(x)\in C_b^\infty (\R^3;\R)$ satisfies $|\nabla s(x)|^2\leq \theta(x,z)$, with,
 $$
\theta (x,z):=\min \{\lambda_0+\delta_0-\re z , \inf\sigma (\re\hat\Pi (Q(x)-z)\hat\Pi)\}-\frac12\delta_1,
$$
then,
$e^{s(x)/h}\widetilde A_\mu(z)e^{-s(x)/h}$ can be written as,
\be
\label{pseudo}
e^{s(x)/h}\widetilde A_\mu(z)e^{-s(x)/h}= \Lambda_{\mu,s} (z)  +L_{\mu,s}(z)\psi_0+\Theta_{\mu,s}(z)
\ee
where $\Lambda_{\mu,s} (z)$ is a $m\times m$ matrix of pseudodifferential operators on $\R^3$, $X_{1,s}:=e^{s(x)/h}X_1e^{-s(x)/h}$, $M_{1,s}:=e^{s(x)/h}M_1e^{-s(x)/h}$,
and 
\be
\label{estavecpoids}
\begin{aligned}
&L_{\mu,s}(z)=-R_\mu^+\left((H_{\rm el}^\mu -\lambda_0-\delta_0)(1-\psi_1)+RX_{0,s}(z)R\right)R_\mu^- + \cO(h);\\
&\Theta_{\mu,s}(z)=\cO(h^\infty)\, :	\, L^2(\R^{3+p})\to L^2(\R^3)^{\oplus N},
\end{aligned}
\ee
with $X_{0,s}:=e^{s(x)/h}X_0e^{-s(x)/h}$. The principal symbol of $\Lambda_{\mu,s} (z)$ is of the form, 
 $$
\sigma_{\mu,s}(x,\xi;z):=\left( z-\left( (I+\mu{}^td\omega (x))^{-1}(\xi+i\nabla s(x))\right)^2\right)I_N - \cM_\mu(x),
 $$
where $I_N$ stands for the $N\times N$ identity matrix, and $\cM_\mu (x)$ is a $N\times N$ matrix of smooth functions on $\R^3$ with eigenvalues $\lambda_1(x+\mu\omega (x)),\dots,\lambda_N(x+\mu\omega (x)$ for $x\in \R^3\backslash\cW_0$, and satisfying $\re\cM_\mu (z)\geq \lambda_0+\delta_0$ for $x\in\cW_0$. Finally, there exists a constant $C>0$ such that,
\be
\label{propLmu}
\re e^{s(x)/h}L_{\mu,s}(z)e^{-s(x)/h}\leq Ch.
\ee
Here, $\psi_0$ and $\psi_1$ are the two functions defined as in (\ref{partfunct}).
\end{theorem}
\begin{remark}\sl
Again, using Remark \ref{rmmubar}, we see that this is also equivalent to: $0\in \sigma(\widetilde A_{\overline\mu}(\overline z))$.
\end{remark}
\begin{proof} We have to prove (\ref{pseudo}). Using the same notations as in (\ref{partfunct}), we have,
$$
H_\mu=P=P_1+ (Q-\lambda_0-\delta_0)(1-\psi_1)=P_1+ (Q-\lambda_0-\delta_0)(1-\psi_1)\psi_0,
$$
and thus,
\be
\label{Amutilde}
\widetilde A_\mu(z)=R_\mu^+ \left(z-P_1\right)R_\mu^- +R_\mu^+\left(\widetilde B_\mu(z)-(Q-\lambda_0-\delta_0)(1-\psi_1)\psi_0 \right)R_\mu^-.
\ee
The fact that $R_\mu^+ \left(z-P_1\right)R_\mu^-$ is  a matrix of smooth semiclassical pseudodifferential operators on $\R^3$ is a direct consequence of the fact that $P_1$ is a twisted pseudodifferential operator associated with the family  $(\Omega_j, U_j)_{j=0,\dots, J}$, and that the same is true for $R_\mu^\pm$ (see \cite{MaSo2} for the terminology and details). Moreover, its symbol is a second-order polynomial with respect to $\xi$, and its principal symbol is of the form $z-\left( (I+\mu{}^td\omega (x))^{-1}\xi\right)^2 - \cM_\mu(x) $, where $\cM_\mu(x)$ is the matrix,
$$
\cM_\mu(x):= R_\mu^+\left(Q(x)\psi_1(x)+(\lambda_0+\delta_0)(1-\psi_1(x))\right)R_\mu^-.
$$
In particular, when $x\in  \R^3\backslash\cW_0$, then $\psi_1(x)=1$, and the eigenvalues of $\cM_\mu(x)$ are those of $Q(x)\Pi_{\rm el}^\mu(x)$, that is, $\lambda_1(x+\mu\omega (x)),\dots,\lambda_N(x+\mu\omega (x)$. Moreover, when $x\in\cW_0$, then $\re Q(x)\geq \lambda_0+\delta_0$, and thus $\re \cM_\mu (x)\geq \lambda_0+\delta_0$, too. As a consequence, $e^{s/h}R_\mu^+ \left(z-P_1\right)R_\mu^-e^{-s/h}$ is a pseudodifferential operator, too, with principal symbol $z-\left( (I+\mu{}^td\omega (x))^{-1}(\xi+i\nabla s)\right)^2 - \cM_\mu(x)$.

In view of (\ref{Amutilde}), and since $R_\mu^+(Q-\lambda_0-\delta_0)(1-\psi_1)\psi_0 R_\mu^-$ commutes with $e^{s/h}$, now we are reduced to study $R_\mu^+e^{s/h}\widetilde B_\mu(z)e^{-s/h}R_\mu^-$. We first prove,
\begin{lemma}\sl
If $s=s(x)\in C_b^\infty (\R^3;\R)$ satisfies $|\nabla s(x)|^2\leq \theta(x,z)$, then,
for $j=0,1$, one has,
$$
\re e^{s(x)/h}X_j(z) e^{-s(x)/h}\geq 0,
$$
and,
$$
e^{s(x)/h}X_j(z) e^{-s(x)/h}=\cO(1)\, :\, L^2(\R^3 ; L^2(\R^p)) \to H^2(\R^3 ; L^2(\R^p)),
$$
uniformly as $h\to 0_+$. Moreover, $e^{s(x)/h}X_1(z) e^{-s(x)/h}$ is a twisted pseudodifferential operator associated with the family  $(\Omega_j, U_j)_{j=0,\dots, J}$.
\end{lemma}
\begin{proof}
We have,
$$
e^{s(x)/h}\hat\Pi (P_j-z) \hat\Pi  e^{-s(x)/h}=\hat\Pi (e^{s(x)/h}K_0e^{-s(x)/h}+Q_j-z) \hat\Pi,
$$
and the principal symbol $k_s$ of the semiclassical differential operator\\ $e^{s(x)/h} K_0 e^{-s(x)/h}$ is given by
$$
k_s(x,\xi)=\left((I+\mu {}^td\omega (x))^{-1}(\xi + i\nabla s(x))\right)^2,
$$
and thus, for $\mu$ small enough,
$$
\re k_s(x,\xi)=(1+\cO(|\mu|))\xi^2-(1+\cO(|\mu|))|\nabla s|^2\geq \frac12\xi^2-\theta (x,z)-\frac18\delta_1.
$$
As a consequence, for $h>0$ small enough,  we obtain,
\be
\label{ellipWeight}
\begin{aligned}
\re e^{s(x)/h}\hat\Pi (P_j-z) \hat\Pi  e^{-s(x)/h}&\geq \re \hat\Pi (-\frac12h^2\Delta_x+Q_j-\theta(x,z)-\frac14\delta_1)\hat\Pi\\
& \geq \re \hat\Pi (-\frac12h^2\Delta_x+\frac14\delta_1)\hat\Pi.
\end{aligned}
\ee
Since $e^{s(x)/h}X_j(z) e^{-s(x)/h} = \left(e^{s(x)/h}\hat\Pi (P_j-z) \hat\Pi  e^{-s(x)/h}\left|_{{\rm Ran}\hat\Pi}\right.\right)^{-1}\hat\Pi$, the first two results follow. Moreover, for $j=1$, we know that $e^{s(x)/h}\hat\Pi (P_1-z) \hat\Pi  e^{-s(x)/h}$ is a twisted pseudodifferential operator associated with the family  $(\Omega_j, U_j)_{j=0,\dots, J}$. Thus, so is $e^{s/h}\cG_1(z)e^{-s/h}$, and (\ref{ellipWeight}) shows that it is elliptic. Then, the last result follows from the general theory of \cite{MaSo2}.
\end{proof}

This lemma allows us to extend the result of Theorem \ref{abstractth} by taking into account the weight $e^{s/h}$. Indeed, working with $e^{s/h}\cG_j(z)e^{-s/h}$ instead of $\cG_j(z)$, we see that all the arguments can be repeated, the main point being that, in this case,  $Q(x)$ will be substituted with $Q(x)-\left( (I+\mu{}^td\omega (x))^{-1}\nabla s\right)^2$, and $K_0$ will be $\left( (I+\mu{}^td\omega (x))^{-1}(D_x+i\nabla s)\right)^2+\left( (I+\mu{}^td\omega (x))^{-1}\nabla s\right)^2$, leaving satisfied the conditions (\ref{K0pos})-(\ref{assabstr}).

In particular, 
according to the expression of $\widetilde B_\mu (z)$ and the definition of $Y$ given in (\ref{defbasic}), and adding an index $s$ to the operators, meaning that they are conjugated with $e^{s/h}$ (when they don't commute with it), we have,
$$
\widetilde B_{\mu,s} (z)=e^{s/h}\widetilde B_\mu (z)e^{-s/h}=-Y_{2,s}+\left(M_{0,s} -Y_{4,s}+B_0T_{0,s}+B_1T_{1,s}\right)(1-Y_{1,s})
$$
with,
$$
B_0, B_1=\cO(h)\, :\, H^{-1}(\R^3 ; L^2(\R^p)) \to L^2(\R^3 ; L^2(\R^p)).
$$
Thus, using the expressions of $Y_2$ and $Y_1$ given in (\ref{defbasic}) and the definition of $R$,
$$
\begin{aligned}
\widetilde B_{\mu,s} (z)=&-\varphi_1M_{1,s}X_{1,s}M_{1,s}\varphi_1\\
&+\left(M_{0,s} -Y_{4,s}+B_0'T_{0,s}+B'_1T_{1,s}\right)(1-\varphi_1X_{1,s}M_{1,s}X_{1,s})+B_2\psi_0,
\end{aligned}
$$
with $B_2=-RX_{0,s}R +\cO(h)$.

Now, since $T_0$ and $T_1$ are differential operators with coefficients supported in $\{\psi_0=1\}$, and since $\Pi (M_0-Y_4)\Pi =0$, we deduce,
$$
\begin{aligned}
R_\mu^+\widetilde B_{\mu, s}(z)R_\mu^-=&\Lambda_{\mu,1} -R_\mu^+(M_{0,s} -Y_{4,s}+B_0'T_{0,s}+B'_1T_{1,s})\varphi_1X_{1,s}M_{1,s}R_\mu^-\varphi_1\\
&+R_\mu^+B_2'R_\mu^-\psi_0,
\end{aligned}
$$
with
$$
\begin{aligned}
&\Lambda_{\mu,1}:=-\varphi_1R_\mu^+M_{1,s}X_{1,s}M_{1,s}R_\mu^-\varphi_1;\\
&B_2'=-RX_{0,s}R +\cO(h).
\end{aligned}
$$
By the arguments of \cite{MaSo2}, we see that $\Lambda_{\mu,1}$ is a semiclassical pseudodifferential operator of order $-2$ (and thus $\cO(h^2)$ on $L^2(\R^3)$). Moreover, since $[Q,\Pi]\varphi_1 =0$, we have $M_0\varphi_1 = [K_0,\Pi]\varphi_1=M_1\varphi_1$, and $Y_4\varphi_1 =\varphi_0M_1\varphi_0\varphi_1+\varphi_1M_1\varphi_1^2$. as a consequence (still with the arguments of \cite{MaSo2}), we see that $R_\mu^+(M_{0,s} -Y_{4,s})\varphi_1X_{1,s}M_{1,s}R_\mu^-\varphi_1$ is a semiclassical pseudodifferential operator of order $-2$, too. Hence, we are led to,
$$
R_\mu^+\widetilde B_{\mu,s}(z)R_\mu^-=\Lambda_{\mu,1}' -R_\mu^+(B_0'T_{0,s}+B'_1T_{1,s})\varphi_1X_{1,s}M_{1,s}R_\mu^-\varphi_1+R_\mu^+B_2'R_\mu^-\psi_0,
$$
where $\Lambda_{\mu,1}' $ is a semiclassical pseudodifferential operator of order $-2$. We prove,
\begin{lemma}\sl
For $k=1,2$, one has,
$$
T_{k,s}X_{1,s}M_{1,s}(1-\psi_0)=\cO(h^\infty)\, :\, L^2(\R^{3+p})\to L^2(\R^{3+p}).
$$
\end{lemma}
\begin{proof} We write $X_{1,s}M_{1,s}$ by using the general expression of a twisted pseudodifferential operator (see \cite{MaSo2}, Definition 4.4), namely,
$$
X_{1,s}M_{1,s}=\sum_{j=0}^J U_j^{-1}\chi_j A_j^NU_j\chi_j + \cO(h^N),
$$
where $N\geq 1$ is arbitrary, $\chi_j\in C_b^\infty (\R^3)$ is supported in $\Omega_j$, and $A_j^N$ is a pseudodifferential operator (of degree -1) with respect to the variable $x$, with operator-valued symbol. Then, the result immediately follows from the fact that $T_{k,s}$ is a first-order differential operator with smooth coefficients supported in $\{\psi_0 =1\}$.
\end{proof}

Then, formula (\ref{pseudo}) follows by setting,
\be
\label{Theta}
\begin{aligned}
 \Lambda_{\mu,s}:=&R_\mu^+ \left(z-P_{1,s}\right)R_\mu^- + \Lambda_{\mu,1}';\\
L_{\mu,s}:=& R_\mu^+ (B_2'-(B_0'T_{0,s}+B'_1T_{1,s})\varphi_1X_{1,s}M_{1,s}R_\mu^-\varphi_1\\
&-(Q-\lambda_0-\delta_0)(1-\psi_1))R_\mu^-;\\
 \Theta_{\mu,s}(z):=& -R_\mu^+(B_0'T_{0,s}+B'_1T_{1,s})\varphi_1X_{1,s}M_{1,s}R_\mu^-\varphi_1(1-\psi_0).\\
\end{aligned}
\ee
\end{proof}

\section{Action on WKB solutions}
\label{sect5}

In this section, we study the action of $\widetilde A_\mu (z)$ on WKB functions of the type $a(x,h)e^{-s(x)/h}$, where the symbol $a(x,h)$ admits some semiclassical expansion of the type,
$$
a(x,h)\sim \sum_{k\geq 0}h^{k/2}a_k(x)
$$
as $h\to 0_+$. We first prove,
\begin{proposition}\sl
\label{asyX1}
Let $u\in L^2(\R^{3+p})$ such that, for all $j=0,1,\dots,r$, and $x\in\Omega_j$, $U_j(x)u(x,y)$ that can be written as $U_j(x)u(x,y)=a_j(x,y;h)e^{-s(x)/h}$ with $s\in C_b^\infty (\R^3;\R)$ independent of $h$, $|\nabla s(x)|^2\leq \theta(x,z)$,  $a_j\in C^\infty (\Omega_j; L^2(\R^p))$, $a_j$ admits in $C^\infty (\Omega_j; L^2(\R^p))$  an asymptotic expansion of the type,
$$
a_j(x,y;h)\sim \sum_{k\geq 0}h^{k/2}a_{j,k}(x,y)
$$
as $h\to 0_+$, with $a_{j,k}(x,y)\in C^\infty (\Omega_j; L^2(\R^p))$ independent of $h$. Then, for any $\chi_j\in C_0^\infty (\Omega_j)$, the function $e^{s/h}U_j\chi_jX_1(z)u$ admits in $C^\infty (\Omega_j; H^2(\R^p))$ an asymptotic expansion of the type,
$$
e^{s(x)/h}U_j\chi_jX_1(z)u(x,y;h)\sim \sum_{k\geq 0}h^{k/2}b_{j,k}(x,y;z),
$$
with,
$$
b_{j,0} := U_j[\hat\Pi\left(Q_1 (x)-[(1+\mu{}^td\omega)^{-1}\nabla s]^2-z\right)\hat\Pi]^{-1}\hat \Pi U_j^{-1} \chi_ja_{j,0}.
$$
Moreover, the support in $x$ of $b_{j,k}$ is included in the union of the supports in $x$ of $\chi_ja_{j,0}, \chi_ja_{j,1},\dots, \chi_ja_{j,k}$.
\end{proposition}
\begin{proof} Let $\tilde\chi_j\in C_0^\infty (\Omega_j)$ such that $\tilde\chi_j \chi_j=\chi_j$. As we have seen in the previous section, $e^{s/h}\cG_1(z)^{-1}e^{-s/h}$ is an elliptic twisted pseudodifferential operator.
Thus, by the general theory of \cite{MaSo2} (in particular Propositions 4.6, 4.10 and 4.14), we know that $U_j\tilde\chi_je^{s/h}\cG_1(z)^{-1}e^{-s/h} U_j^{-1}\tilde\chi_j$ is a bounded $h$-admissible operator on $L^2(\R^3 ; L^2(\R^p \oplus \C^m)$ (with operator-valued symbol). In particular, $U_j\tilde\chi_je^{s/h}X_1(z)e^{-s/h} U_j^{-1}\tilde\chi_j$ is a bounded $h$-admissible operator on $L^2(\R^3 ;L^2(\R^p))$. This means that, for any $M\geq 1$, it can be written as,
$$
\begin{aligned}
U_j\tilde\chi_je^{s/h}X_1(z)e^{-s/h}& U_j^{-1}\tilde\chi_j v(x,y)\\
&=\frac1{(2\pi h)^n}\iint  e^{i(x-x')\xi /h}q_M(x,\xi)v(x',y)dx'd\xi +R_Mv,
\end{aligned}
$$
with $||R_M||=\cO(h^M)$ and $q_M$ is an operator-valued symbol, acting on $L^2(\R^p)$. The symbolic calculus also gives us,
$$
q_M=q_0  +\cO(h),
$$
with,
$$
q_0(x,\xi):=U_j\tilde\chi_j\hat\Pi\left([(1+\mu{}^td\omega)^{-1}(\xi + i\nabla s)]^2+Q_1 (x)-z\right)^{-1}\hat\Pi U_j^{-1}\tilde\chi_j.
$$

In particular, taking $v=a$, and observing that $e^{-s/h}$ commutes with $U_j^{-1}\tilde\chi_j$, we obtain,
$$
\begin{aligned}
U_j\tilde\chi_je^{s/h}X_1(z)&\tilde\chi_j u(x,y)\\
&=\frac1{(2\pi h)^n}\iint e^{i(x-x')\xi /h}q_N(x,\xi)a(x',y)dx'd\xi +\cO(h^N).
\end{aligned}
$$
The stationary-phase theorem applied to this integral (with critical point $\xi =0$ and $x'=x$) immediately gives,
\be
\label{asyX1}
U_j\tilde\chi_je^{s/h}X_1(z)\tilde\chi_j u(x,y)\sim \sum_{k\geq 0}h^{k/2}\tilde b_{j,k}(x,y;z),
\ee
in $C^\infty (\Omega_j; H^2(\R^p))$, with,
$$
\tilde b_{j,0} := U_j[\hat\Pi\left(Q_1 (x)-[(1+\mu{}^td\omega)^{-1}\nabla s]^2-z\right)\hat\Pi]^{-1}\hat \Pi U_j^{-1} \tilde\chi_j^2a_{j,0},
$$
and ${\rm Supp}\, b_{j,k}\subset \cup_{0\leq\ell\leq k}{\rm Supp}\, \tilde\chi_j a_{j,\ell}$.
On the other hand, using the representation of twisted $h$-admissible operators given in \cite{MaSo2}, Definition 4.4 (that we apply to $e^{s/h}X_1(z)e^{-s/h}$), and still by the stationary phase theorem, we also have,
$$
U_j\chi_je^{s/h}X_1(z)(1-\tilde\chi_j )u(x,y)\sim 0
$$
in $C^\infty (\Omega_j; H^2(\R^p))$. As a consequence, we have,
$$
\begin{aligned}
U_j\chi_je^{s/h}X_1(z)u(x,y) &\sim U_j\chi_je^{s/h}X_1(z)\tilde\chi_j u(x,y)\\
& \sim\chi_jU_j\tilde\chi_je^{s/h}X_1(z)\tilde\chi_j u(x,y),
\end{aligned}
$$
and the result follows from (\ref{asyX1}).
\end{proof}
\begin{corollary}\sl
Let $u\in (L^2(\R^{3}))^{\oplus N}$ that can be written as $u(x;h)=a(x;h)e^{-s(x)/h}$ with $s\in C_b^\infty (\R^3;\R)$ independent of $h$, $|\nabla s(x)|^2\leq \theta(x,z)$,  $a\in (C^\infty (\R^3))^N$ admitting  an asymptotic expansion of the type,
$$
a(x;h)\sim \sum_{k\geq 0}h^{k/2}a_k(x)
$$
as $h\to 0_+$, with $a_k(x)\in (C^\infty (\R^3))^N$, ${\rm Supp}\, a_k \subset \cW_0^c$. Then, $e^{s/h}\widetilde A_\mu(z)u$ admits a semiclassical asymptotic expansion of the type,
$$
e^{s/h}\widetilde A_\mu(z)u(x;h)\sim  \sum_{k\geq 0}h^{k/2}b_k(x;z),
$$
with,
$$
b_0(x;z):=\sigma_{\mu ,s}(x,0;z)a_0=\left( z+\left( (I+\mu{}^td\omega (x))^{-1}\nabla s(x)\right)^2\right)a_0 - \cM_\mu(x)a_0,
$$
and ${\rm Supp}\, b_k(\cdot ;z) \subset \cW_0^c$.
\end{corollary}
\begin{proof}
Indeed, by Proposition \ref{asyX1} and (\ref{wkmutwist}), we see that, for any $j=0,\dots,r$ and $\chi_j\in C_0^\infty (\Omega_j)$, we have,
$$
U_j\chi_jX_{1,s}M_{1,s}R_\mu^-\varphi_1 a(x,y;h)\sim \sum_{k\geq 0}h^{k/2}b_{j,k}(x,y;z),
$$
with ${\rm Supp}\, b_{j,k} \subset \cW_0^c\times \R^p$. In particular,
\be
\label{pres0}
\psi_0 U_j\chi_jX_{1,s}M_{1,s}R_\mu^-\varphi_1 a(x,y;h) =\cO(h^\infty),
\ee
together with all its derivatives in $x$. Now, by (\ref{pseudo}), and using that $\psi_0u=0$, we have,
$$
e^{s(x)/h}\widetilde A_\mu(z)u= \Lambda_{\mu,s} (z)a  +\Theta_{\mu,s}(z)a,
$$
where $\Theta$ is given in (\ref{Theta}), and is of the form
$$
\Theta_{\mu,s}(z)=L_{\mu,s}'(z)\varphi_1X_{1,s}M_{1,s}R_\mu^-\varphi_1(1-\psi_0)
$$
Looking more carefully at the expression of $L_{\mu,s}'(z)$, we see that it involves only twisted pseudodifferential operators, except $X_0$ that always appears on the form $X_0T_0$ (this is due to the expressions of the operators $Y$ and $Y_3$ given in (\ref{defbasic}), and to the fact that $L_{\mu,s}'(z)$ does not involve $Y_1$, $Y_2$ nor $Y_4$). Therefore, since the coefficients of $T_0$ are supported in $\{ x\not =0\}$, and the same holds for $\varphi_1$, one can deduce from (\ref{pres0}) and the general theory of \cite{MaSo2} that one has,
$$
\Theta_{\mu,s}(z)a=\cO(h^\infty),
$$
together with all its derivatives. As a consequence, we obtain,
$$
e^{s(x)/h}\widetilde A_\mu(z)u\sim \Lambda_{\mu,s} (z)a,
$$
and the result follows from the fact that $\Lambda_{\mu,s}$ is a pseudodifferential operator with principal symbol $\sigma_{\mu,s}(x,\xi ;z)$, and from a standard stationary-phase expansion.
\end{proof}

\section{Location of resonances}
\label{sectellipt}

In order to determine the resonances of $H$ near $\lambda_0$, we see on (\ref{feschb}) that it is necessary to know the spectrum of $\widetilde A_\mu(z)$ near $0$. But we see on (\ref{pseudo}) with $s=0$ that $\widetilde A_\mu(z)$ is not really a pseudodifferential operator (not even modulo ${\mathcal O}(h^\infty)$), because of the term $L_{\mu,0}(z)\psi_0$ in its expression. However, because this term is  localised in the region where $\re \Lambda_{\mu,0}\leq -\delta_0+Ch$, and has a real part $\leq C'h$ (with $C,C'$ positive constants), there exists a general method, due to Helffer and Sj\"strand \cite{HeSj2}, to compare the eigenvalues of $\widetilde A_\mu(z)$ to those of,
$$
\hat  A_\mu(z):= \Lambda_{\mu,0}(z)+\Theta_{\mu,0}(z)=\widetilde A_\mu(z)- L_{\mu,0}(z).
$$
Assume for instance that $\hat  A_\mu(z)$ admits an isolated eigenvalue $\rho_0=\rho_0(z)$ close to $0$ (with normalised eigenfunction $u_0$), call $\Pi_0$ the spectral projector of $\hat  A_\mu(z)$ associated with $\rho_0$, $\hat\Pi_0:= 1-\Pi_0$, and assume that the reduced resolvent $\hat\Pi_0(\hat  A_\mu(z)-\rho)^{-1}\hat\Pi_0$ is not exponentially large for $\rho$ in a small ($h$-dependent) neighbourhood of $0$. Then, for $\rho$ is his neighbourhood, one considers the two Grushin problems,
$$
{\mathcal G_0}(\rho):=\left(\begin{array}{cc}
\hat A_\mu-\rho & u_0\\
\la \cdot,u_0\ra & 0
\end{array}
\right);
$$
$$
{\mathcal G}(\rho):=\left(\begin{array}{cc}
\widetilde A_\mu-\rho & \chi_1u_0\\
\la \cdot,u_0\ra & 0
\end{array}
\right),
$$
where we have omitted the variable $z$, and where $\chi_1\in C^\infty (\R^n)$ is such that $\chi_1(x)=1$ for $|x|\geq 3\delta_1$, $\chi_1(x)=0$ for $|x|\leq 2\delta_1$. Then, by construction ${\mathcal G_0}(\rho)$ is invertible, and we denote its inverse by,
$$
{\mathcal G_0}(\rho)^{-1}=\left(\begin{array}{cc}
E_0(\rho) & E_0^+(\rho)\\
E_0^-(\rho) & E_0^{-+}(\rho)
\end{array}
\right),
$$
then a candidate for the inverse of ${\mathcal G}(z)$ is given by (see also \cite{HeSj2}, Formula (9.22)),
$$
{\mathcal F}(\rho):=\left(\begin{array}{cc}
\chi_1 E_0(\rho)\chi_2+(B-\rho)^{-1}(1-\chi_2) & \chi_1E_0^+(\rho)\\
E_0^-(\rho) & E_0^{-+}(\rho)
\end{array}
\right),
$$
where $\chi_2\in C^\infty (\R^n)$ is such that $\chi_1(x)=1$ for $|x|\geq 4\delta_1$, $\chi_1(x)=0$ for $|x|\leq 3\delta_1$, and $B$ is defined as,
$$
B:=\widetilde A_\mu-C\chi_3,
$$
with $\chi_3\in C^\infty (\R^n)$ such that $\chi_1(x)=1$ near $\{ \re {\mathcal M}_\mu (x)\leq \lambda_0\}$,  $\chi_1(x)=0$ for $|x|\leq \delta_2$ (where $\delta_2 >4\delta_1$), and $C>0$ is taken sufficiently large in order that $\re B\leq -\delta_0$. 

Indeed, taking advantage of the fact that $u_0$ is exponentially small near $0$, and that, thanks to (\ref{estavecpoids}), the operators satisfy the same type of estimates as in \cite{HeSj2}, Proposition 9.3, we see as in \cite{HeSj2}, Section 9, that one has,
$$
{\mathcal G}(\rho){\mathcal F}(\rho)=I+{\mathcal O}(e^{-\alpha_0 /h}),
$$
where $\alpha_0>0$ mainly depends on the distance between the support of $\chi_3$ and 0. By a similar procedure, an approximate left-inverse of ${\mathcal G}(\rho)$ can also be found, and as in \cite{HeSj2}, this permits to show the existence of an eigenvalue $\rho_1(z)$ of $\widetilde A_\mu(z)$ exponentially close to $\rho_0(z)$. Finally, the corresponding resonance of $H$ is obtained by solving the equation $\rho_1(z)=0$, and we see on Theorem \ref{Thpseudo} that this leads to a unique value $z_1$ close to $\lambda_0$.

As in \cite{HeSj2}, this argument can also be extended to a set of resonances of $\hat A_\mu$ separated from the rest of its spectrum.

\section{Applications}
\label{sect7}
In this section we discuss some applications to cases where resonances can be located quite well, and estimates on their widths can be obtained. We do not give details on the proofs (that may result rather long) but just give indications on them.

\subsection{Shape resonances}

In this subsection we assume that $N=1$, and that the first electronic level $\lambda_1(x)$ presents, at some energy $\lambda_0$, the geometric situation of a well in an island, as described in \cite{HeSj2}, that is,
\begin{itemize}
\item There is an open bounded connected set $\ddot O$ (the island) and a compact set $U\subset \ddot O$ (the well) such that $\lambda_1\leq\lambda_0$ on $U\cup \ddot O^c$, $\lambda_1 > \lambda_0$ on $\ddot O\backslash U$;
\item The set $\ddot O^c\times \R^3$ is non trapping for the Hamiltonian $\xi^2 + \lambda_1(x)$ at energy level $\lambda_0$;
\item $\lambda_1(x)$ admits a limit $\lambda_1^\infty < \lambda_0$ as $|x|\to\infty$, $x$ in a complex sector of the form $\{ |\im x| < \delta |\re x|\}$ with $\delta >0$.
\end{itemize}
\begin{remark}\sl Note that the limit $|x|\to\infty$ in $H_{\rm el}(x)$ can be deduced, by a change of variable, from the semiclassical limit of $-\tilde h^2\Delta_y + W(y)$, with $\tilde h:= |x|^{-\frac12}$, and $W(y):=\alpha + \sum_{j} (\alpha_j^+ |y_j+\theta|^{-1} + \alpha_j^- |y_j-\theta|^{-1}) + \sum_{j,k}\alpha_{jk}|y_j-y_k|^{-1}$, where $\theta$ is any element of the unit sphere of $\R^3$. In particular, the limit $\lambda_1^\infty$ can be seen to exist for $p=1$, and to be equal to the smallest between the first eigenvalue of $-\Delta_y + \alpha_1^+|y|^{-1}$, and that of $-\Delta_y + \alpha_1^-|y|^{-1}$. The case $p\geq 2$ seems to be more delicate to treat, but it is reasonable to think that the limit should exist, too.
\end{remark}
In this situation, we can adapt some of the arguments of \cite{HeSj2} (see also \cite{LaMa} for a more simplified version) to the operator $\Lambda_{\mu,0}(z)$ given in (\ref{pseudo}) for $s=0$. Moreover, the properties of $L_{\mu,s}$ and $L'_{\mu,s}$ (in particular (\ref{propLmu})) allows us to extend the Agmon estimates appearing in \cite{HeSj2} to the whole operator $\widetilde A_\mu (z)$. 

In addition, the non degeneracy of $\lambda_1(x)$ and the rotational symmetry of $H_{\rm el}(x)$ (namely, that rotating simultaneously $x$ and $y_j$ ($j=1,\dots,p$) with the same rotation of $\R^3$, leaves $H_{\rm el}(x)$ unchanged), one can see as in \cite{KMSW}, Theorem 2.1 (see also \cite{GKMSS}) that one can construct the functions $ w_{k,\mu}$ in such a way that $\widetilde A_\mu (z)$ commutes with the operator of angular momentum with respect to $x$. In particular, working in polar coordinates $(r,\theta)\in \R_+\times S^2$, denoting by $Y_{\ell,m}$ the spherical harmonic of degree $\ell$ and order $m$, and $\cH_\ell$ the subspace of $L^2(S^2)$ spanned by $\{Y_{\ell,m}\, ;\, |m|\leq \ell\}$, one can decompose $\widetilde A_\mu (z)$ as (see, e.g., \cite{So}),
$$
\widetilde A_\mu (z) =\bigoplus_{\ell \geq 0}\widetilde A_\mu^{\ell} (z)\otimes {\mathbf 1}_{{\mathcal H}_\ell},
$$
where  the action of $\widetilde A_\mu^ {\ell} (z)$ on $L^2(\R_+, r^2dr)$ is defined by,
$$
\widetilde A_\mu^ {\ell} (z)\alpha (r):=\la \widetilde A_\mu (z)\alpha (r)Y_{\ell,m} (\theta), Y_{\ell,m}(\theta)\ra_{L^2(S^2)},
$$
where actually the right-hand side does not depend on $m$.

Then, taking the cutoff functions $\varphi_0, \varphi_1, \psi_0, \psi_1$ radial, we see on (\ref{pseudo}) (with $s=0$) that $\widetilde A_\mu^{\ell} (z)$ can be written as,
$$
\widetilde A_\mu^{\ell} (z)= z - F_\mu^{\ell}(z) - S_\mu^{\ell}(z)\psi_0-T_\mu^{\ell}(z),
$$
with $T_\mu^{\ell}(z)={\mathcal O}(h^\infty)$, $\re S_\mu^{\ell}(z)\leq Ch$, and $ F_\mu^{\ell}(z)$ given by,
$$
F_\mu^{\ell}(z)u (r):=\la F_\mu (z)u (r)Y_{\ell,m} (\theta), Y_{\ell,m}(\theta)\ra_{L^2(S^2)},
$$
where $F_\mu (z)$ is a rotational-invariant semiclassical pseudodifferential operator on $L^2(\R^3)$, with  symbol $f_\mu (x,\xi;z)$ satisfying,
$$
f_\mu (x,\xi )=[(I+\mu^td\omega(x))^{-1}\xi]^2+\lambda_\mu (x)+{\mathcal O}(h^2),
$$
where $\lambda_\mu$ is a smooth function of $|x|$ such that,
$$
\begin{aligned}
& \lambda_\mu (|x|)=\lambda_1(x+\mu\omega (x))\quad \mbox{for } |x|\geq 2\delta_1;\\
& \re \lambda_\mu (x)\geq \lambda_0+\delta_0\quad \mbox{for } |x|\leq 2\delta_1.
\end{aligned}
$$
Setting $v(r)=ru(r)$, this leads to a problem on $L^2(\R^+, dr)$ with Dirichlet boundary condition at 0, and with principal part,
$$
P_\mu = -h^2\left(\frac{d}{(1+\tilde\omega'(r))dr}\right)^2+\lambda_\mu(r).
$$
Then, thanks also to (\ref{estavecpoids}), one can adapt the arguments of \cite{HeSj2}, Section 9, to the operator $F_\mu^{\ell}(z) +T_\mu^{\ell}(z)$, and, by the method described in Section \ref{sectellipt} one can show the existence of resonances of $H$ near $\lambda_0$, with exponentially small widths.

\subsection{Microlocal tunneling} Here we assume $N=2$, and that the second electronic level $\lambda_2(x)$ forms a well at some energy $\lambda_0$, while the first one $\lambda_1(x)$ is non trapping at $\lambda_0$. More precisely, we assume that $\lambda_2(x)$ is simple (so that $\lambda_1$ and $\lambda_2$ are automatically rotationally invariant), and that,

\begin{itemize}
\item The set $U:=\{ \lambda_2\leq\lambda_0\}$ is compact;
\item $\lambda_2(x)$ admits a  limit $\lambda_2^\infty  >\lambda_0$ as $|x|\to\infty$, $x$ in a complex sector of the form $\Gamma_\delta:=\{ |\im x| < \delta |\re x|\}$ with $\delta >0$;
\item $\lambda_1(x)$ admits a  limit $\lambda_1^\infty  <\lambda_0$ as $|x|\to\infty$, $x\in \Gamma_\delta$;
\item The set $\{ \lambda_1(x)=\lambda_0\}$ is reduced to a single point. 
\end{itemize}

In this case, using again the rotational symmetry and the simplicity of $\lambda_1$ and $\lambda_2$, the operator $\widetilde A_\mu (z)$ can be written as,
$$
\widetilde A_\mu (z) =\bigoplus_{\ell \geq 0}\widetilde A_\mu^{\ell} (z)\otimes {\mathbf 1}_{{\mathcal H}_\ell},
$$
where, as before, $\cH_\ell$ is the subspace of $L^2(S^2)$ spanned by the spherical harmonics $Y_{\ell,m}$ ($-\ell\leq m\leq \ell$), and  $\widetilde A_\mu^ {\ell} (z)$ is a $2\times2$ matrix  acting on $L^2(\R_+, r^2dr)\oplus L^2(\R_+, r^2dr)$, of the form,
$$
\widetilde A_\mu^{\ell} (z)= z - F_\mu^{\ell}(z) - S_\mu^{\ell}(z)\psi_0-T_\mu^{\ell}(z),
$$
still with $T_\mu^{\ell}(z)={\mathcal O}(h^\infty)$, $\re S_\mu^{\ell}(z)\leq Ch$, and $ F_\mu^{\ell}(z)$ is given by,
$$
F_\mu^{\ell}(z)u (r):=\la F_\mu (z)u (r)Y_{\ell,m} (\theta), Y_{\ell,m}(\theta)\ra_{L^2(S^2)},
$$
where $F_\mu (z)$ is a rotational-invariant $2\times 2$ matrix of semiclassical pseudodifferential operators, with  symbol $f_\mu (x,\xi;z)$ satisfying,
$$
f_\mu (x,\xi )=[(I+\mu^td\omega(x))^{-1}\xi]^2{\mathbf I}_2+\Lambda_\mu (x)+{\mathcal O}(h^2),
$$
where $\Lambda_\mu$ is a $2\times 2$ matrix-valued smooth function such that,
$$
\begin{aligned}
& \Lambda_\mu (x)=\left(\begin{array}{cc}
\lambda_1(x+\mu\omega (x)) & 0\\
0 & \lambda_2(x+\mu\omega (x))
\end{array}\right)
\quad \mbox{for } |x|\geq 2\delta_1;\\
& \re \Lambda_\mu (x)\geq \lambda_0+\delta_0\quad \mbox{for } |x|\leq 2\delta_1.
\end{aligned}
$$

In this situation, one can work in the same spirit as in \cite{Ma3} (but in a simpler way, here, since the operator is already distorted, and thus only compactly supported weights are necessary) and prove the existence of resonances near $\lambda_0$ with exponentially small widths as $h\to 0_+$. Alternatively, one can also adapt the arguments of \cite{FMW}, where particular solutions to such a system are constructed.

\subsection{Molecular predissociation} In this subsection we take $N=3$, and we assume that the second and third level cross on some disc $\{ |x| =r_0\}$. More precisely, we assume that the first 3 eigenvalues can be re-indexed in such a way that they become smooth functions of $r=|x|$, and that they satisfy,
\begin{itemize}
\item The set $U:=\{ r>0\, ;\, \lambda_2(r)\leq\lambda_0\}$ is a bounded interval $[r_1,r_2]$;
\item $\lambda_2(r)$ admits a  limit $\lambda_2^\infty  >\lambda_0$ as $r\to\infty$, $r$ in a complex sector of the form $\Gamma_\delta:=\{ |\im r| < \delta |\re r|\}$ with $\delta >0$;
\item $\lambda_3(r)$ admits a  limit $\lambda_3^\infty  <\lambda_0$ as $r\to\infty$, $r\in \Gamma_\delta$;
\item $\lambda_1(r)$ admits a  limit $\lambda_1^\infty  <\lambda_3^\infty$ as $r\to\infty$, $r\in \Gamma_\delta$;
\item The set $\{ \lambda_3(r)=\lambda_0\}$ is reduced to a single point $\{ r_3\}$ with $r_3>r_2$;
\item The set $\{ \lambda_1(r)=\lambda_0\}$ is reduced to a single point belonging to $(0,r_1)$;
\item For all $r>0$, $\lambda_1(r)<\min\{ \lambda_2(r), \lambda_3(r)\}$.
\end{itemize}
Then, following \cite{Kl}, one can prove the existence of resonances near $\lambda_0$ with exponentially small widths as $h\to 0_+$.

\subsection{Crossing levels} We take again $N=3$, and we assume that the second and third level cross on some disc $\{ |x| =r_0\}$. More precisely, we assume that the first 3 eigenvalues can be re-indexed in such a way that they become smooth functions of $r=|x|$, and that they satisfy,
\begin{itemize}
\item The set $U:=\{ r>0\, ;\, \lambda_2(r)\leq\lambda_0\}$ is a bounded interval $[r_1,r_2]$ with $r_1<r_2$ and $\lambda_2'(r_j)\not= 0$ for $j=1,2$;
\item $\lambda_2(r)$ admits a  limit $\lambda_2^\infty  >\lambda_0$ as $r\to\infty$, $r$ in a complex sector of the form $\Gamma_\delta:=\{ |\im r| < \delta |\re r|\}$ with $\delta >0$;
\item $\lambda_3(r)$ admits a  limit $\lambda_3^\infty  <\lambda_0$ as $r\to\infty$, $r\in \Gamma_\delta$;
\item $\lambda_1(r)$ admits a  limit $\lambda_1^\infty  <\lambda_3^\infty$ as $r\to\infty$, $r\in \Gamma_\delta$;
\item The set $\{ \lambda_3(r)=\lambda_0\}$ is reduced to $\{ r_2\}$;
\item The set $\{ \lambda_1(r)=\lambda_0\}$ is reduced to a single point belonging to $(0,r_1)$;
\item For all $r>0$, $\lambda_1(r)<\min\{ \lambda_2(r), \lambda_3(r)\}$.
\end{itemize}
Then, one can adapt the arguments of \cite{FMW} (see, in particular, Remarks 2.2 and  8.8 in \cite{FMW}) and prove the existence of resonances at a distance ${\mathcal O}(h^{2/3})$ of $\lambda_0$, with  widths of size $h^{5/3}$.


\end{document}